\documentclass[pdflatex,sn-mathphys-num]{sn-jnl}

\usepackage{graphicx}%
\usepackage{multirow}%
\usepackage{amsmath,amssymb,amsfonts}%
\usepackage{amsthm}%
\usepackage{mathrsfs}%
\usepackage[title]{appendix}%
\usepackage{xcolor}%
\usepackage{textcomp}%

\usepackage{booktabs}%
\usepackage{algorithm}%
\usepackage{algorithmicx}%
\usepackage{algpseudocode}%
\usepackage{listings}%

\theoremstyle{thmstyleone}%
\newtheorem{theorem}{Theorem}[section]
\newtheorem{proposition}[theorem]{Proposition}%

\theoremstyle{thmstyletwo}%
\newtheorem{remark}{Remark}%

\theoremstyle{thmstylethree}%

\begin{document}

\title[Calculation of solitons derived from quivers]
{Calculation of solitons derived from quivers}


\author{\fnm{Takashi} \sur{Ichikawa}}\email{ichikawn@cc.saga-u.ac.jp}




\affil{\orgdiv{Department of Mathematics}, 
\orgname{Faculty of Science and Engineering}, 
\orgaddress{\street{Saga 840-8502}, \country{Japan}}}





\abstract{
We associate each quiver with soliton solutions of nonlinear integrable systems 
containing the KP and Toda hierarchies. 
We give an explicit and computable formula of these soliton solutions which are regarded 
as universal ones obtained by degenerating quasi-periodic solutions. 
}




\maketitle

\section{Introduction}

In this paper, 
we associate each quiver, i.e., oriented graph, 
with soliton solutions of nonlinear integrable systems 
containing the Kadomtsev-Petviashvili (KP) and Toda lattice hierarchies. 
Furthermore, we give an explicit and computable formula of these soliton solutions 
which are regarded  as universal ones obtained as degenerations (or tropical limits) of 
quasi-periodic solutions given by Krichever \cite{Kr}, Mumford \cite{Mum1, Mum2} and others. 

For the KP case, 
our quiver soliton derived from an oriented graph $\Delta$ is given in terms of 
the associated $\tau$-function 
$$
\tau(\mbox{\boldmath $t$}) = \sum_{v \in D} 
\exp \left( \frac{1}{2} v B_{0} v^{T} + \left( c + \sum_{n=1}^{\infty} r_{n} t_{n} \right) v^{T} \right) 
\eqno(1.1) 
$$
of $\mbox{\boldmath $t$} = (t_{1}, t_{2},...)$. 
Here $c$ is any vector, 
$D$ is a Delaunay (finite) set in $H_{1}(\Delta, \mathbb{Z})$, 
and $B_{0}, r_{n}$ are expressed by  parameters attached to oriented edges in $\Delta$ 
as the periods, differentials for degenerate algebraic curves with dual graph $\Delta$ respectively. 
Although $\tau(\mbox{\boldmath $t$})$ is defined in \cite[Theorem 5.3]{I23} 
as the regularized limit of shifted Riemann theta functions for families of Riemann surfaces 
defined by a length function $l$ on $\Delta$, 
we show that $\tau(\mbox{\boldmath $t$})$ is independent of the choice of $l$, 
and can be explicitly calculated from only $\Delta$. 
We also give exact formulas of $B_{0}, r_{n}$ by modifying results of \cite[Section 5]{I23}. 
Furthermore, if $\Delta$ is planar, 
then $\tau(\mbox{\boldmath $t$}) > 0$ for real variables $t_{1}, t_{2},...$, 
and hence it gives rise to (real and) regular solutions of the KP (II) equation 
$$
3 \partial_{y}^{2} u = 
\partial_{x} \left( 4 \partial_{t} u - \partial_{x}^{3} u - 6 u \, \partial_{x} u \right), 
$$
where $\mbox{\boldmath $t$} = (t_{1} = x, t_{2} = y, t_{3} = t)$. 

Another large class of (regular) KP solitons called {\it line solitons} derived from 
(positive) Grassmannians was especially studied by Kodama 
(cf. \cite{Ko17} and references there in).  
See results of Abenda-Grinevich \cite{AG18, AG19, AG22}, 
Agostini-Fevola-Mandelshtamb-Sturmfels \cite{AgFMS}, Ichikawa-Kodama \cite{IK25}, 
Kodama \cite{Ko24} and Nakayashiki \cite{Nak1, Nak2} 
for relations on between line solitons and degenerations of quasi-periodic solutions. 

The universal $\tau$-functions for Toda solitons with main term similar to (1.1) 
were obtained in \cite[Section 6]{I24}, 
and hence one can describe these functions as quiver solitons. 
We hope that similar formulas will be obtained for universal $\tau$-functions for 
the universal hierarchy introduced by Krichever-Zabrodin \cite{KrZ}.

\section{Solitons associated with tropical curves}

\subsection{Tropical curves and their periods}

We consider a (finite and connected) graph $\Delta = (V, E)$ 
consisting of the sets $V, E$ of vertices and edges respectively. 
Furthermore, we fix an orientation of $\Delta$ which gives two maps 
$s: E \rightarrow V$ and $t: E \rightarrow V$ taking each oriented edge, or arrow, 
to its source and target vertex respectively, 
and hence $(\Delta; s, t)$ becomes a quiver as an oriented graph. 
For each oriented edge $e \in E$, 
let $-e$ be the edge $e$ with inverse orientation, 
and put 
$$
\pm E = \{ \pm e \mid e \in E \}. 
$$ 
Following \cite{MZ}, 
a tropical curve $C = (V, E, l)$ (of weight $0$) is defined as a graph $\Delta = (V, E)$ 
with length function $l : E \rightarrow \mathbb{R}_{> 0}$, 
and its genus $g_{C}$ is defined as ${\rm rank}_{\mathbb{Z}} H_{1}(\Delta, \mathbb{Z})$. 
Denote by $\langle \cdot, \cdot \rangle$ the $\mathbb{Z}$-linear form 
on $\left( \bigoplus_{e \in E} \mathbb{Z} e \right)^{\otimes 2}$ satisfying 
$$
\langle e, e' \rangle = \left\{ \begin{array}{ll} 
l(e)  & (e = e'), \\ 
0     & (e \neq e') 
\end{array} \right. 
$$ 
for $e, e' \in E$. 

We associate a tropical curve $C = (V, E, l)$ with a family of {\it real} curves, 
namely, algebraic curves defined over $\mathbb{R}$, which has genus $g_{C}$. 
Let $R_{0}$ be a semistable real curve with dual graph $\Delta$ which is obtained from 
$P_{v} = \mathbb{RP}^{1}$ $(v \in V)$ by identifying real points 
$$
x_{e} \in P_{s(e)} \setminus \{ \infty \} \ \ \text{and} \ \ 
x_{-e} \in P_{t(e)} \setminus \{ \infty \} \quad (e \in E), 
$$
where $x_{h} \neq x_{h'}$ for $h, h' \in \pm E$ satisfying $h \neq h'$ and $s(h) = s(h')$. 
Then it is shown in \cite[Section 2]{I23} that there exists a family of real curves 
$\{ R_{s} \}_{s \geq 0}$ as deformations of $R_{0}$ by real parameters $y_{e} = y_{-e}$ 
$(e \in E)$ with $|y_{e}| = s^{l(e)}$ by the relation 
$$
\xi_{e} \cdot \xi_{-e} = y_{e}; \quad \xi_{h} := z - x_{h} \ (h \in \pm E) 
\eqno(2.1)
$$
for sufficiently small $s \geq 0$. 
Take a symplectic basis $\{ a_{i}, b_{i} \}_{1 \leq i \leq g}$ of $H_{1}(R_{s}, \mathbb{Z})$ for $s > 0$ such that 
$a_{i}$ tend singular points on $R_{0}$ and $b_{i}$ are taken as oriented loops in $\Delta$ 
which give a basis of  $H_{1}(\Delta, \mathbb{Z})$. 
Then the period matrix $B_{C}$ of $C$ is defined as $(\langle b_{i}, b_{j} \rangle)_{1 \leq i, j \leq g}$ 
(cf. \cite{MZ}). 

It is shown in \cite[Theorems 3.4 (1) and 5.1--5.3]{I23} that 
there exists a unique family $\{ \omega_{1},..., \omega_{g} \}$ of basis consisting of 
regular stable differentials, 
namely continuous global sections of the dualizing sheaf (cf. \cite[Section 1]{DeM}) 
on $\{ R_{s} \}_{s \geq 0}$ which are normalized for $\{ a_{i} \}_{i}$ in the sense that 
$$
\mbox{$\displaystyle \frac{1}{2 \pi \sqrt{-1}} \oint_{a_{j}} \omega_{i}$ 
is the Kronecker delta $\delta_{ij}$ for $s > 0$.}
$$ 
Furthermore, 
the period matrices 
$\displaystyle B_{s} = \left( \oint_{b_{i}} \omega_{j} \right)_{1 \leq i,j \leq g}$ 
of $(R_{s}, \{ a_{i}, b_{i} \})$ satisfy the variational formula 
$$
B_{s} = \log s \cdot B_{C} + B_{0}
$$
for a certain $g \times g$ symmetric matrix $B_{0}$.

\subsection{Solitons in the KP system}

Let the notation be as above. 
For $\alpha \in \mathbb{R}^{g}$, 
denote by 
$$
\theta_{C}(\alpha) = 
\max \left\{ \left. \alpha B_{C} v^{T} - \frac{1}{2} v B_{C} v^{T} \, \right| 
v \in \mathbb{Z}^{g} \right\} 
$$
its tropical theta function, 
and by 
$$
D_{C, \alpha} = \left\{ v \in \mathbb{Z}^{g} \left| \,\alpha B_{C} v^{T} 
- \frac{1}{2} v B_{C} v^{T} = \theta_{C}(\alpha) \right. \right\} 
\eqno(2.2) 
$$
which is the Delaunay (finite) set consisting of vertices of the Delaunay polytope for $\alpha$ 
given by the metric $|v|_{B_{C}} = v B_{C} v^{T}$ $(v \in \mathbb{R}^{g})$ 
(cf. \cite[Section 18]{Nam}, \cite[Section 2]{AgFMS}). 

For a smooth point $p$ on $R_{0} = \bigcup_{v \in V} P_{v}$; 
$P_{v} = \mathbb{RP}^{1}$ with coordinate $z$, 
put $r_{n} = (r_{i,n})_{1 \leq i \leq g}$, where 
$$
\omega_{i}|_{R_{0}} = \left\{ \begin{array}{ll} 
{\displaystyle \sum_{n=1}^{\infty} r_{i,n} (z - p)^{n-1} dz} \ \ 
\mbox{at $z = p$}  & (p \neq \infty), 
\\ 
{\displaystyle - \sum_{n=1}^{\infty} r_{i,n} z^{-n-1} dz} \ \ 
\mbox{at $z = \infty$} & (p = \infty). 
\end{array} \right.
$$
Then the following result was obtained in \cite[Theorem 5.3]{I23} 
by showing that $\tau_{C, \alpha}(\mbox{\boldmath $t$})$ becomes 
the regularized limit of quasi-periodic KP solutions given by Krichever \cite{Kr} 
associated with the family $\{ R_{s} \}_{s > 0}$ of Riemann surfaces.

\begin{theorem} {\rm (cf. \cite[Theorem 5.3]{I23})} 

{\rm (1)}
The formal power series in $\mbox{\boldmath $t$} = (t_{1}, t_{2},...)$ defined as 
$$
\tau_{C, \alpha}(\mbox{\boldmath $t$}) = 
\sum_{v \in D_{C, \alpha}} \exp \left( \frac{1}{2} v B_{0} v^{T} + 
\left( c + \sum_{n=1}^{\infty} r_{n} t_{n} \right) v^{T} \right) 
\eqno(2.3)
$$
gives a solution of the KP hierarchy \cite{S, DKJM, Z}: 
$$
\oint e^{-2 \sum_{n=1}^{\infty} s_{n} \lambda^{n}} 
\tau(\mbox{\boldmath $t$} - \mbox{\boldmath $s$} - [\lambda^{-1}]) 
\tau(\mbox{\boldmath $t$} + \mbox{\boldmath $s$} + [\lambda^{-1}]) 
\frac{d \lambda}{2 \pi \sqrt{-1}} = 0, 
$$
where $\mbox{\boldmath $s$} = (s_{1}, s_{2}, s_{3},...)$, 
$[\lambda^{-1}] = \left( \lambda^{-1}, \lambda^{-2}/2, \lambda^{-3}/3,... \right)$ and 
$\displaystyle \oint \cdot \frac{d \lambda}{2 \pi \sqrt{-1}}$ 
means taking the residue at $\lambda = \infty$. 

{\rm (2)} 
If $\mbox{\boldmath $t$} = (t_{1}, t_{2}, t_{3})$ with $x = t_{1}$, $y = t_{2}$, $t = t_{3}$, 
then $u(x, y, t) = 2 \partial_{x}^{2} \log \tau(\mbox{\boldmath $t$})$ gives a solution of 
the KP (II) equation: 
$$
3 \partial_{y}^{2} u = 
\partial_{x} \left( 4 \partial_{t} u - \partial_{x}^{3} u - 6 u \, \partial_{x} u \right). 
\eqno(2.4)
$$

\end{theorem}

\subsection{Regularity of solitons}

A Riemann surface $R$ of genus $g$ is called an {\it {\texttt M}-curve} 
if $R$ is a smooth projective curve  defined over $\mathbb{R}$ 
and its real locus consisting of $\mathbb{R}$-rational points on $R$ 
has $g + 1$ connected components. 
The following theorem is a detailed explanation of statements given in \cite[Section 5.2]{I23}. 

\begin{theorem} 
Assume that the graph $\Delta$ underlying a tropical curve $C$ is planar. 
Then we can construct a family $\{ R_{s} \}_{s > 0}$ of Schottky uniformized {\texttt M}-curves 
for which the above $B_{0}$ and $r_{n}$ are real, namely, all their entries are real numbers. 
Therefore, the associated $\tau$-function $\tau_{C,\alpha}(\mbox{\boldmath $t$})$ given in Theorem 2.1 
is positive if $t_{i} \in \mathbb{R}$, 
and hence $u(x, y, t)$ gives a regular solution of (2.4) in real variables $x, y, t$. 
\end{theorem}

\begin{proof} 
Since the above $R_{0}$ is a real curve on the fixed plane obtained from 
$P_{v} = \mathbb{RP}^{1} = \mathbb{R} \cup \{ \infty \}$ $(v \in V)$ by identifying 
$x_{e} \in P_{s(e)}$ and $x_{-e} \in P_{t(e)}$ $(e \in E)$, 
and its deformations $R_{s}$ by the relation $\xi_{e} \cdot \xi_{-e} = -s^{l(e)}$ for 
small $s > 0$ are {\texttt M}-curves. 
As is shown in \cite[Section 3]{I23},  
$\omega_{i}$ $(1 \leq i \leq g)$ are defined as real-valued forms on $R_{s}$ 
which implies that $B_{0}$ and $r_{n}$ are real. 
\end{proof}

\section{Calculation of quiver solitons}

\subsection{Calculation of Delaunay sets}

\begin{proposition} 
Let $\{ b_{i} \}_{i}$ be the basis of $H_{1}(\Delta, \mathbb{Z})$ given in Section 2.1, 
and take another basis $\{ b'_{i} \}_{i}$ represented by 
$(b'_{1},..., b'_{g}) = (b_{1},..., b_{g}) X$ for an element $X$ of ${\rm GL}_{g}(\mathbb{Z})$. 
Then the $\tau$-function defined in Theorem 2.1 (1) corresponding to $\{ b'_{i} \}_{i}$ becomes 
$\tau_{C, \alpha X^{T}}(\mbox{\boldmath $t$})$ given in (2.3). 
\end{proposition}

\begin{proof} 
Take elements $a'_{i}$ $(1 \leq i \leq g)$ of $H_{1}(R_{s}, \mathbb{Z})$ for $s > 0$ such that 
$\{ a'_{i}, b'_{i} \}_{i}$ gives its symplectic basis, 
and denote by $\{ \omega'_{i} \}_{i}$ the basis of regular stable differentials on $R_{0}$ normalized 
for $\{ a_{i} \}_{i}$ which is represented as 
$(\omega'_{1},..., \omega'_{g}) = (\omega_{1},..., \omega_{g}) X$. 
Then 
$$
\left( \oint_{b'_{i}} \omega'_{j} \right)_{i,j} = X^{T} \left( \oint_{b_{i}} \omega_{j} \right)_{i,j} X, 
$$
and hence the $\tau$-function defined in Theorem 2.1 (1) corresponding to $\{ b'_{i} \}_{i}$ becomes 
$$
\sum_{v \in D_{C, \alpha}} \exp \left( \frac{1}{2} v X^{T} B_{0} X v^{T} + 
\left( \sum_{n=1}^{\infty} r_{n} t_{n} \right) X v^{T} \right) 
= \tau_{C, \alpha X^{T}}(\mbox{\boldmath $t$}). 
$$
\end{proof}

By this proposition, 
it is enough to calculate the $\tau$-function (2.3) for a special basis of $H_{1}(\Delta, \mathbb{Z})$. 
Let $C = (\Delta, l)$ be a tropical curve of genus $g$, 
and denote by $\varphi_{\Delta}$ the linear map 
$$
\varphi_{\Delta}: \mathbb{Z}^{g} = H_{1}(\Delta, \mathbb{Z}) \rightarrow 
\mathbb{Z}^{E} = \bigoplus_{e \in E} \mathbb{Z} \cdot e 
\eqno(3.1) 
$$ 
sending $b = h(1) \cdots h(k) \in H_{1}(\Delta, \mathbb{Z})$ to $\sum_{i=1}^{k} h(i)$, 
where $h(i) \in \pm E$. 
Then it is shown by Mumford \cite[Theorem 18.2]{Nam} that the Delaunay decomposition of 
$\mathbb{R}^{g}$ is independent of the length function $l$, 
and that each Delaunay set is expressed as 
$$
\left\{ u \in \mathbb{Z}^{g} \mid \varphi_{\Delta}(u) \in I_{e} \ (e \in E) \right\} 
$$ 
for certain sets $I_{e}$ of a form 
$\{ n_{e} \}$ or $\{ n_{e}, n_{e} + 1 \}$ with $n_{e} \in \mathbb{Z}$. 
Using this result of Mumford, 
we calculate Delaunay sets for special basis of $H_{1}(\Delta, \mathbb{Z})$. 
Take an orientation of $\Delta$ and a maximal subtree $T$ of $\Delta$. 
Then $\Delta \setminus T$ consists of $g$ oriented edges which we denote by $e_{1},..., e_{g}$, 
and hence there exist uniquely homology basis $\{ b_{T,1},..., b_{T,g} \}$ of $\Delta$ 
such that each $b_{T,i}$ consists of $e_{i}$ and the unique path in $T$ from $t(e_{i})$ to $s(e_{i})$. 
We identify 
$$
\mathbb{Z}^{g} = \bigoplus_{i=1}^{g} \mathbb{Z} \cdot b_{T,i}, 
\quad \mathbb{R}^{g} = \bigoplus_{i=1}^{g} \mathbb{R} \cdot b_{T,i}. 
$$

\begin{proposition}
Let $D_{C, \alpha}$ denote the Delaunay set defined in (2.2) for the tropical curve 
$C = (\Delta, l)$ and $\alpha \in \mathbb{R}^{g}$. 
Then $D_{C, \alpha}$ is a subset of $\bigoplus_{i=1}^{g} D_{i}$, 
where each $D_{i}$ is represented by $\{ m_{i} \}$ or $\{ m_{i}, m_{i}+1 \}$ 
for a certain integer $m_{i}$, 
which consists of elements $d \in \bigoplus_{i=1}^{g} D_{i}$ satisfying that for the map $\varphi_{\Delta}$ 
given in (3.1), 
$\varphi_{\Delta}(d) \in \bigoplus_{e \in E} I_{e}$, 
where each $I_{e}$ is represented by $\{ n_{e} \}$ or $\{ n_{e}, n_{e}+1 \}$ for a certain integer $n_{e}$. 
\end{proposition}

\begin{proof}
For each positive number $s$, 
we define the length function $l_{s}$ on the set $E$ of edges in $\Delta$ by 
$$
l_{s}(e) = \left\{ \begin{array}{ll} 
1 & (e \in \{ e_{1},..., e_{g} \}), \\ s \cdot l(e) & (e \not\in \{ e_{1},..., e_{g} \}). 
\end{array} \right. 
$$
Then the associated period matrix $B_{s}$ is decomposed as $E_{g} + s B$, 
where $E_{g}$ denotes the unit matrix of degree $g$, 
and $B$ is a symmetric and nonnegative definite matrix. 
By the result of Mumford, 
$D_{C, \alpha}$ becomes a Delaunay set for the metric associated with $B_{s}$ 
for any $s > 0$. 
Since any Delaunay set for the standard metric associated with 
$E_{g} = \lim_{s \downarrow 0} B_{s}$ is given by $\bigoplus_{i=1}^{g} D_{i}$, 
where $D_{i} = \{ m_{i} \}$ or $\{ m_{i}, m_{i}+1 \}$ for integers $m_{i}$, 
$D_{C, \alpha}$ is contained in  $\bigoplus_{i=1}^{g} D_{i}$ for certain $D_{i}$'s. 
Therefore, the assertion follows from Mumford's result. 
\end{proof}

\subsection{Calculation of tropical differentials and periods}

For each $i = 1,..., g$, let $b_{i}$ and $\omega_{i}$ be as in Section 2.1 such that 
$b_{i}$ is conjugate to the product $h_{i}(1) \cdots h_{i}(k_{i})$ of elements of $\pm E$ 
which is reduced in the sense that $h_{i}(k_{i}) \neq -h_{i}(1)$. 
Denote by $C_{v}$ the irreducible component of $R_{0}$ corresponding to $v \in V$. 
then by \cite[Theorem 3.3 (1)]{I23}, 
the pullback $(\omega_{i}|_{C_{v}})^{*}$ of the restriction $\omega_{i}|_{C_{v}}$ to $P_{v}$ is represented as 
$$
\left( \omega_{i}|_{C_{v}} \right)^{*} = 
\sum_{s(h_{i}(m)) = v} \frac{dz}{z - x_{h_{i}(m)}} - \sum_{t(h_{i}(n)) = v} \frac{dz}{z - x_{-h_{i}(n)}}. 
$$ 
Therefore, the integer 
$$
a(i, h) := \left| \{ m \in \{ 1,..., k_{i} \} \mid h_{i}(m) = h \} \right| - 
\left| \{ n \in \{ 1,..., k_{i} \} \mid -h_{i}(n) = h \} \right| 
\eqno(3.2)
$$
satisfies 
$$
\left( \omega_{i}|_{C_{v}} \right)^{*} = \sum_{s(h) = v} a(i, h) \frac{dz}{z - x_{h}}. 
$$
Then we have the following theorem from \cite[Theorems 5.1 and 5.3]{I23} 
by correcting the formula of $B_{0}$ given \cite[p. 1719]{I23}.

\begin{theorem} 
Let $a(i, h)$ $(1 \leq i \leq g, h \in \pm E)$ be the integers defined in (3.2).  

{\rm (1)} 
Fix a vertex $v_{0}$ of $V$ and a coordinate $s = z^{-1}$ at $\infty \in P_{v_{0}} = \mathbb{CP}^{1}$. 
Then the above $r_{n}$ is given by 
$$
r_{n} = \left( \sum_{s(h) = v_{0}} -a(i, h) x_{h}^{n} \right)_{1 \leq i \leq g}. 
$$

{\rm (2)} 
For $i, j = 1,..., g$, denote by $(B_{0})_{ij}$ the $(i, j)$-entry of the matrix $B_{0}$. 
If $b_{i}$ (resp. $b_{j}$) are conjugate to the reduced products $h_{i}(1) \cdots h_{i}(k_{i})$ 
(resp. $h_{j}(1) \cdots h_{j}(l_{j})$) of elements of $\pm E$, 
then 
\begin{align*}
\lefteqn{
\exp \left( (B_{0})_{ij} \right) \cdot \prod_{m=1}^{l_{j}} {\rm sgn} (y_{h_{j}}(m)) 
} 
\\
& = 
\prod_{m=1}^{l_{j}} \prod_{s(h) = s(h_{j}(m))} \left[ \prod_{h \neq h_{j}(m)} (x_{h_{j}(m)} - x_{h})^{a(i, h)} 
\prod_{h \neq -h_{j}(m-1)} (x_{-h_{j}(m-1)} - x_{h})^{-a(i, h)} \right],
\end{align*}
where $y_{h}$ are the real parameters given in (2.1), 
and $h_{j}(0) = h_{j}(l_{j})$. 
In particular, 
$$
\exp \left( (B_{0})_{ii} \right) = 
\prod_{m=1}^{l_{i}} \left[ \frac{-{\rm sgn} (y_{h_{i}}(m))}{(x_{h_{j}(m)} - x_{-h_{j}(m-1)})^{2}} \right]. 
$$ 

{\rm (3)} 
If $\Delta$ is planar, 
then all $(B_{0})_{ij}$ are real numbers, 
and hence 
\begin{align*}
\lefteqn{
\exp \left( (B_{0})_{ij} \right) 
} 
\\
& = 
\prod_{m=1}^{l_{j}} \prod_{s(h) = s(h_{j}(m))} \left[ \prod_{h \neq h_{j}(m)} 
|x_{h_{j}(m)} - x_{h}|^{a(i, h)} 
\prod_{h \neq -h_{j}(m-1)} |x_{-h_{j}(m-1)} - x_{h}|^{-a(i, h)} \right] 
\end{align*}
and 
$$
(B_{0})_{ii} = -2 \sum_{m=1}^{l_{i}} \log |x_{h_{j}(m)} - x_{-h_{j}(m-1)}|. 
$$  
\end{theorem}

\begin{proof} 
The assertion (1) follows from \cite[Theorem 5.3]{I23} since 
$$
\frac{dz}{z - x_{h}} = - s^{-1} (1 + x_{h} s + (x_{h} s)^{2} + \cdots ) ds. 
$$
We show the assertion (2). 
The family $\{ R_{s} \}_{s > 0}$ of Riemann surfaces is obtained from $\bigcup_{v \in V} P_{v}$ 
by identifying $z_{e} \in P_{s(e)}$ and $z_{-e} \in P_{t(e)}$ by the relation 
$z_{e} \cdot z_{-e} = y_{e}$ with $|y_{e}| = s^{l(e)}$, 
where $\lim_{s \rightarrow 0} z_{h} = x_{h}$ $(h \in \pm E)$. 
Then $\displaystyle \oint_{b_{i}} \omega_{j}$ is asymptotically for $s \rightarrow 0$ 
a sum with integral coefficients of the following integrals   
$$
\int_{a}^{z_{h}} \frac{dz}{z - x_{k}} \ \ 
\left( s(h) = s(k), \, h \neq k, \, a \in P_{s(h)} \right), 
$$
$$
\int_{z_{-h}}^{b} \frac{dz}{z - x_{k}} \ \ 
\left( s(-h) = s(k), \, -h \neq k, \, b \in P_{t(h)} \right) 
$$
and 
$$
\int_{a}^{z_{h}} \frac{dz}{z - x_{h}} - \int_{z_{-h}}^{b} \frac{dz}{z - x_{-h}} \ \ 
\left( a \in P_{s(h)}, \, b \in P_{t(h)} \right). 
$$
Since 
$$
\lim_{s \rightarrow 0} \exp \left( \int_{a}^{z_{h}} \frac{dz}{z - x_{k}} \right) = 
\frac{x_{h} - x_{k}}{a - x_{k}}, 
\ \ 
\lim_{s \rightarrow 0} \exp \left( \int_{z_{-h}}^{b} \frac{dz}{z - x_{k}} \right) = 
\frac{b - x_{k}}{x_{-h} - x_{k}} 
$$
and 
$$
\exp \left( \int_{a}^{z_{h}} \frac{dz}{z - x_{h}} - \int_{z_{-h}}^{b} \frac{dz}{z - x_{-h}} \right) = 
\frac{y_{h}}{(a - x_{h})(b - x_{-h})}, 
$$
the assertion (2) follows from the definition of $B, B_{0}$. 
Furthermore, (3) follows from (2) and Theorem 2.2. 
\end{proof}

\subsection{Calculation of quiver solitons} 
By applying Propositions 3.1, 3.2 and Theorem 3.3 to the calculation of 
$\tau_{C, \alpha}$ given in Theorem 2.1, 
we give the following explicit and computable formula of KP soliton solutions in (2.3) 
in terms of quivers $\Delta$. 

\begin{theorem}
The $\tau$-function for the KP soliton solution in (2.3) is described 
in terms of a quiver $(\Delta; s, t)$ of genus $g$ with maximal subtree $T$ as 
$$
\tau(\mbox{\boldmath $t$}) = 
\sum_{v \in D} \exp \left( \frac{1}{2} v B_{0} v^{T} + 
\left( c + \sum_{n=1}^{\infty} r_{n} t_{n} \right) v^{T} \right). 
\eqno(3.3)
$$
Here $D$ is a finite subset of $\mathbb{Z}^{g}$ described in Proposition 3.2 
for the basis $\{ b_{T,1},..., b_{T,g} \}$ associated with $T$, 
and $B_{0}, r_{n}$ are given in Theorem 3.3. 
In particular, $u(x, y, t) = 2 \partial_{x}^{2} \log \tau(\mbox{\boldmath $t$})$ 
with $\mbox{\boldmath $t$} = (x, y, t)$ gives a solution of the KP (II) equation (2.4). 
\end{theorem}

\begin{remark}
Since $B_{0}$ is symmetric, 
in order to calculate $\exp \left( v B_{0} v^{T}/2 \right)$, 
it is enough to know the values of $\exp((B_{0})_{ij})$ for $i \neq j$ and of $\exp((B_{0})_{ii}/2)$. 
When $\Delta$ is not planar, 
one can calculate  $\exp((B_{0})_{ii}/2)$ by the way in the proof of Theorem 3.3 (2). 
\end{remark}

\subsection{Example for the banana quivers}
Let $\Delta$ be the $g$-loop banana quiver which consists of $2$ vertices $v_{0}, v_{1}$ and 
$g+1$ oriented edges $e_{0}, e_{1},..., e_{g}$ such that $s(e_{i}) = v_{0}, t(e_{i}) = v_{1}$ $(0 \leq i \leq g)$. 
Take a maximal subtree $T$ consisting $e_{0}$ with boundary vertices $v_{0}, v_{1}$. 
Then $b_{i} = e_{i} \cdot (-e_{0})$ $(1 \leq i \leq g)$, 
and hence by Proposition 3.2, 
the Delaunay set $D \subset \mathbb{Z}^{g}$ is either $\{ (m_{i})_{1 \leq i \leq g} \}$ or 
$$
\{ (a_{i})_{1 \leq i \leq g} \mid \text{$a_{j} = m_{j}, \, m_{j}+1$ and $a_{i} = m_{i}$ for $i \neq j$} \} 
\ \ (1 \leq j \leq g) 
$$ 
for some integers $m_{i}$. 
Since 
$$
a(i, e_{0}) = -1, \ a(i, e_{j}) = \delta_{ij}, \, a(i, -e_{j}) = -\delta_{ij} \ (1 \leq i, j \leq g), 
$$
by Theorem 3.3, we have 
$$
r_{n} = \left( x_{e_{0}}^{n} - x_{e_{i}}^{n} \right)_{1 \leq i \leq g} 
$$
and for $1 \leq i, j \leq g$ with $i \neq j$, 
\begin{align*}
(B_{0})_{ii} & = 
-2 \log \left| (x_{e_{i}} - x_{e_{0}})(x_{-e_{0}} - x_{-e_{i}}) \right|, 
\\
(B_{0})_{ij} & = 
\log \left| \frac{(x_{e_{j}} - x_{e_{i}})(x_{-e_{j}} - x_{-e_{i}})}
{(x_{e_{0}} - x_{e_{i}})(x_{-e_{0}} - x_{-e_{i}})(x_{e_{j}} - x_{e_{0}})(x_{-e_{j}} - x_{-e_{0}})} \right|. 
\end{align*}
Then by Theorem 3.4, 
(3.3) gives rise to the quiver soliton associated with $(\Delta, T)$ 
for any vector $c = (c_{i})_{1 \leq i \leq g}$.

\section{Concluding Remarks}
In this paper, we give an explicit and computable formula (3.3) of KP soliton solutions 
which are regarded  as universal ones obtained as degenerations of quasi-periodic solutions. 
On these quiver solitons, 
$B_{0}$ and $r_{n}$ are rational functions of parameters attached to oriented edges 
in given quivers which are regarded as the regularized period matrices and 
the coefficients of differentials for the degenerate algebraic curves respectively. 

As for regular quiver solitons for the KP (II) equation, 
it is very interesting to study their structure, 
e.g., the wave patterns and the degenerating behaviors from quasi-periodic solutions. 
The formula (3.3) of the $\tau$-functions $\tau(\mbox{\boldmath $t$})$ for quiver solitons 
has a similar form to line KP solitons studied by Kodama and other researchers 
(see \cite{Ko17} and references there in) 
since they both are linear finite sums of exponentials of linear functions of $t_{i}$. 
Therefore, based on results in \cite{Ko17} it seems that if $\mbox{\boldmath $t$} = (x, y, t)$ 
is a system of real parameters, 
then the peaks of the quiver soliton $2 \partial_{x}^{2} \log \tau(\mbox{\boldmath $t$})$ 
at fixed $t$ are approximately described by the soliton graph 
$$
\{ (x, y) \mid \mbox{there exist $v_{1}, v_{2} \in D$ such that $v_{1} \neq v_{2}$ 
and $\Theta_{v_{1}} = \Theta_{v_{2}} = \max_{v \in D} \Theta_{v}$} \}, 
$$
where 
$$
\Theta_{v}(x, y, t) = (r_{1} v^{T}) x + (r_{2} v^{T}) y + (r_{3} v^{T}) t  \quad (v \in D). 
$$ 

Finally, we hope to extend the formula of quiver solitons for Krichever-Zabrodin's universal hierarchy. 
\vspace{2ex}

\noindent 
{\bf Declaration of competing interest} 
\vspace{1ex}

The author declares that he has no known competing financial interests 
or personal relationships that could have appeared to influence the work 
reported in this article. 
\vspace{2ex}

\noindent 
{\bf Data availability} 
\vspace{1ex}

No data was used for the research described in this article. 
\vspace{2ex}

\noindent 
{\bf Acknowledgments} 
\vspace{1ex}

This work is partially supported by the JSPS Grant-in-Aid for 
Scientific Research No. 25K06920.

\renewcommand{\refname}{\normalsize{\bf References}}

\end{document}